\RequirePackage{etoolbox}
\csdef{input@path}{%
 {sty/},
 {img/},
 {bib/}
}

\documentclass[numbers,compress]{vmsta2}

\volume{}
\issue{6}
\pubyear{}
\articletype{research-article}

\newtheorem{thm}{Theorem}
\newtheorem{remark}{Remark}

\newtheorem{cor}{Corollary}
\newtheorem{prop}{Proposition}

\theoremstyle{definition}

\usepackage{color}

\newcommand{\be}{\begin{equation}}
\newcommand{\ee}{\end{equation}}

\newcommand{\PP}{\mathbb{P}}
\newcommand{\RR}{\mathbb{R}}
\newcommand{\NN}{\mathbb{N}}
\newcommand{\CC}{\mathbb{C}}

\newcommand{\lLambda}{\widetilde{\Lambda}}

\newcommand{\lPsi}{\widetilde{\Psi}}
\newcommand{\lpsi}{\widetilde{\psi}}
\newcommand{\lPsiM}{\widetilde{\Psi}_{\rm{M}}}
\newcommand{\lpsiM}{\widetilde{\psi}_{\rm{M}}}
\newcommand{\e}{{\rm e}}

\newcommand{\LambdaM}{{\Lambda_{\rm{M}}}}
\newcommand{\lLambdaM}{\widetilde{\Lambda}_{\rm{M}}}

\definecolor{gp}{rgb}{0.0,0.0,0.0} 
\newcommand{\gp}[1]{{\textcolor{gp}{#1}}}
\definecolor{ck}{rgb}{0.0,0.0,0.0} 
\newcommand{\ck}[1]{{\textcolor{ck}{#1}}}

\begin{document}
\begin{frontmatter}

\pretitle{Research Article}

\title{First-return time in fractional kinetics}

\begin{aug}
\author[a]{\inits{M.}\fnms{Marcus}~\snm{Dahlenburg}
}
\author[a,b]{\inits{G.}\fnms{Gianni}~\snm{Pagnini}\thanksref{cor2}
\ead[label=e2]{gpagnini@bcamath.org}\orcid{0000-0001-9917-4614}}
\thankstext[type=corresp,id=cor2]{Corresponding author.}

\address[a]{
\institution{BCAM} -- Basque Center for Applied Mathematics,\\
Alameda de Mazarredo 14, 48009 Bilbao, 
Basque Country -- \cny{Spain}}

\address[b]{\institution{Ikerbasque} -- Basque Foundation for Science,\\
Plaza Euskadi 5, 48009 Bilbao, Basque Country -- \cny{Spain}}


\markboth{M. Dahlenburg, G. Pagnini}{First-return time in fractional kinetics}
\end{aug}

\begin{abstract}
The first-return time 
is the time that it takes \gp{a random walker to go back to the 
initial position for the first time}. 
We study the first-return time when random walkers perform
fractional kinetics, specifically fractional diffusion,
that is modelled within the framework of the
continuous-time random walk on homogeneous space in the uncoupled formulation
\gp{with Mittag--Leffler distributed waiting-times}. 
We consider both Markovian and non-Markovian \gp{settings,} as well as
any kind of symmetric jump-size distributions, 
namely with finite or infinite variance. 
We show that the first-return time density is indeed independent
of the jump-size distribution when it is \gp{symmetric,}
and therefore it is affected only by 
the waiting-time distribution that embodies the memory of the process.
We perform our analysis \gp{in two cases}: 
{\it first jump then wait} and {\it first wait then jump},
and we provide \gp{several exact results,} including 
the relation between results in the 
\gp{Markovian and non-Markovian settings and the} difference between the two cases.
\end{abstract}



\end{frontmatter}

\section{Introduction}
\label{sec:intro}
Fractional kinetics \cite{sokolov_etal-pt-2002,klafter_etal-pw-2005} 
concerns particles' motion 
whose random displacements follow  
a probability density function that evolves
according to equations including fractional calculus operators.
In a sloppy way, we can say that fractional kinetics concerns
processes that follow a 
non-local generalisation of the Fokker--Planck equation.
A noteworthy example is fractional diffusion \cite{mainardi_etal-fcaa-2001}.

The emerging evolution equation in fractional diffusion
generalises the heat diffusion equation by replacing 
the first derivative in time and/or
the second derivative in space by the corresponding fractional derivatives
\cite{gorenflo_etal-cism-1997}.
As a consequence,
the underlying random walk model is not the classical random walk
leading to the Brownian motion but its generalisation 
characterised by power-law distributions. 
Such generalised random walk is the continuous-time random walk (CTRW)
\cite{montroll_etal-jmp-1965} 
that extends classical random walk by replacing a fixed time-step
with a random interval between consecutive jumps.
These random intervals are called waiting-times and,
in the basic formulation \cite{montroll_etal-jmp-1965}, 
are assumed to be i.i.d. random variables and independent
of the jump-sizes, which also are assumed to be i.i.d. random variables.
\ck{When the distribution of the waiting-times displays a power-law,
then the evolution equation of the displacement density function 
contains time-fractional derivatives. 
On the other side, when the
distribution of the jump-sizes displays a power-law, 
then such evolution equation contains space-fractional derivatives.} 
In this second case, the corresponding random walk is
called L\'evy flight \cite{chechkin_etal-anotrans-2008}.  
When the jump-size distribution displays power-law then
the variance of the jumps and that of the walker's density function
are infinite.
When the waiting-time distribution displays a power-law and 
time-fractional derivatives appear, then the process is 
non-Markovian. The process is striclty Markovian when
the waiting-time distribution is exponential 
\cite{zwanzig-jsp-1983,mainardi_etal-pa-2000},
by the way this constrain can be relaxed and consider as
Markovian all those processes with finite mean waiting-times
and non-Markovian those with infinite mean waiting-times.
Here, we study the first-return time when both 
waiting-time and jump-size distribution may display power-law. 

The first-return time is the time that it takes a random walker
to come back for the first time to the starting location.
This kind of problem has a quite important application in the
animal kingdom because it is linked to concepts as
site fidelity, breeding, social associations, optimal foraging 
\cite{greenwood-ab-1980}.
Mathematically, it is strongly related with the first-passage time problem
\cite{redner-2001}.
The main result in first-passage time problems is
the celebrated Sparre Andersen theorem \cite{sparre-andersen-1954b},
which concerns the survival probability on the half-line
\ck{conditioned on the walker's initial position}
\gp{for symmetric Markovian random walks with fixed time-steps
in the presence of an absorbing boundary.}
\ck{Its importance lays on the fact that it states that
the conditional survival probability 
is independent of the jump-size distribution 
when the walker's trajectory starts from the location
of the absorbing boundary, thus, it is universal.} 
As a matter of fact, we show in this paper 
that this universality of first-passage time problems
is reflected also in first-return time problems. 
Therefore, the results \gp{derived here} are indeed
independent of the tails of the walkers' density function
and so the process can be governed by a Gaussian or by a L\'evy stable density,
as well.  
On the other side, the system is affected indeed by memory effects.
Therefore, the results display different behaviour between 
Markovian\ck{, namely with exponentially distributed waiting-times
\cite{zwanzig-jsp-1983,mainardi_etal-pa-2000},}
and non-Markovian frameworks. 

In particular, 
we report here that in discrete-time random walks,
which are Markovian by construction,
when a nearest-neighbour jump-law is taken into account 
then the probability mass function of the first-return time 
is linearly proportional to that of the first-passage time
\cite{redner-2001,kostinski_etal-ajp-2016,gruda_etal-jsm-2025}.
As it is discussed in the following, see formula \eqref{discrete},
this linear proportionality holds also 
when the assumed jump-law is different of 
the nearest-neighbour rule.
However, since the general framework of the CTRW is \gp{considered here},
we study in addition the \gp{specific non-Markovian setting 
generated by Mittag--Leffler distributed waiting-times}
and we determine and investigate also the difference 
emerging between Markovian and non-Markovian formulations.

The formulation of the CTRW is dependent on the choice about
the starting instant for measuring the elapsed time.
In fact, the mesurement of the duration of the first-return time 
can start synchronised with the first jump,
and we label this case as the case {\it first jump then wait}, 
or it can start independently of the process such that
the first jump occurs after the first random waiting-time,
and we label this case as the case {\it first wait then jump}. 
We consider both cases, because, 
even if formulae relating Markovian and non-Markovian frameworks are equal, 
exact results in the Markovian setting differ between the two cases,
and so in the non-Markovian setting\gp{. Moreover,}
we also quantify this difference.

We derive exact results in the Markovian setting and
we provide the relation between the Markovian and the non-Markovian setting
such that non-Markovian results 
\gp{for Mittag--Leffler distributed waiting-times}
are also exact and fully determined.
The derivation of the result is based on the Sparre Andersen theorem
and it exploits the Laplace transform method. 
In particular, we provide some preliminary notions and results in
Section \ref{sec:preliminaries} where the Sparre Andersen theorem 
is reminded and we also briefly remind the theory of 
CTRW \cite{montroll_etal-jmp-1965}.
This allows us to state the Sparre Andersen theorem in the framework of
the CTRW.
In this section, 
we also derive integral \eqref{SAintegral},
that is the main formula for the calculation of the remaining results.
In Section \ref{sec:FRT}, first we formulate the problem of
first-return time and later we derive exact results both in
the case {\it first jump then wait} and {\it first wait then jump}.  

\section{Preliminaries}
\label{sec:preliminaries}
\subsection{The Sparre Andersen theorem}
Let $\mathbb{R}$ be the set of real numbers, 
we denote by $\mathbb{R}^+$ and by $\mathbb{R}_0^+$
the set of positive real numbers 
and the set of non-negative real numbers, e.g.,
the half-line of spatial excursions excluded the origin
$\mathbb{R}^+ = \{x \in\mathbb{R}|x > 0\}$ and 
the elapsed times
$\mathbb{R}_0^+ = \{t \in \mathbb{R}| t \ge 0\}$.
\ck{Analogously, we denote by $\mathbb{R}^-$ and 
by $\mathbb{R}_0^-$ the set of negative real numbers 
and the set of non-positive real numbers.} Let furthermore 
$\mathbb{N}$ be the set of natural numbers and $\NN_0=\NN\cup\{0\}$.
Let $S_n=X_1 + \dots + X_n$ be the sum of $n\in \NN$ 
i.i.d. random variables $X_n \in \RR$.
Then the first ladder epoch 
$\{\mathcal{T}=n\}=\{S_1 \ge 0, \dots, S_{n-1} \ge 0, S_n < 0\}$ 
is the epoch of the first entry of the walker
into the negative semi-axis $\RR^-$.
By adding a constant $x_0\in\RR_0^+$ to all terms, 
we obtain a random walk with initial position at $x_0$.  
The probability for a walker started at $S_0=x_0$ to remain
in the initial half-axis after $n$ steps,
that is the probability for $\mathcal{T}$ to be larger than $n$, 
is called survival probability \ck{conditioned on the initial position}
and the site $x=0$ is termed as the location of an absorbing barrier.
Thus, \gp{let the survival probability 
conditioned on the initial position
be denoted by $\phi_n:\RR_0^+\to[0,1]\,,\forall n\in\NN$,
then for a walker started at $S_0=x_0$
we have $\phi_n(x_0)=\PP(\mathcal{T} > n)$} and it yields
\be
\int_0^{+\infty} k(z-x_0) \, \phi_n(z) \, dz = \phi_{n+1}(x_0) \,,
\quad x_0 \in \RR_0^+\,, \quad n \in \NN_0 \,,
\label{WH}
\ee
with initial condition $\phi_0(\xi)=1\,,\forall\,\xi \in \RR_0^+$, 
where $k: \RR \to \RR_0^+$ is the probability density function of the 
i.i.d. random variables $X_n$, $\forall \, n\in\NN$. 
Within the framework of the CTRW,
the density function $k(x)$ \gp{turns into} the jump-size distribution.

Integral equation \eqref{WH} is known as Wiener--Hopf  
equation \cite{estrada_kanwal-2000}. 
A first connection between random walks with an absorbing barrier and
Wiener--Hopf integrals was pointed out by D.V. Lindley in 1952 
\cite{lindley-1952}. 
For a more general relations between
Wiener--Hopf integrals and probabilistic problems,
see, e.g., \ck{reference} \cite[Sections XII.3a and XVIII.3]{feller-1971}.
However, as \ck{observed} by W. Feller,
the connections between first-passage problems and the
Wiener--Hopf integrals {\it "are not as close as they are usually made
to appear"} \cite[Introduction to Chapter XII]{feller-1971}.
For more recent papers, see, e.g., references 
\cite{ivanov-aa-1994,
frisch_etal-1994,
majumdar_etal-jsp-2006,
majumdar-pa-2010,
bray_etal-ap-2013,
majumdar_etal-jpa-2017,dahlenburg_etal-jpa-2022}. 
A general solution to the Wiener--Hopf equation \eqref{WH}
is known in \ck{the} literature
with the name of Pollaczek--Spitzer formula \cite{bray_etal-ap-2013}.
This name refers to a formula, 
see \cite[theorem 5, formula (4.6)]{spitzer-1957} and
\cite[formula (12)]{bray_etal-ap-2013}, 
derived by F. Spitzer in 1957 \cite[theorem 3, formula (3.1)]{spitzer-1957}
on the basis of an auxiliary formula obtained by F. Pollaczek in 1952 
\cite[formula (8)]{pollaczek-1952} but through a different method. 
It is possible to show that problem \eqref{WH} is
equivalent to a Sturm--Liouville problem with proper 
boundary conditions \cite{dahlenburg_etal-prsa-2023},
and this is also an alternative and easier approach 
with respect to the Pollaczek--Spitzer formula for the calculation  
of the survival probability.
In particular, one of the \gp{required} boundary conditions 
for the Sturm--Liouville system emerged to be given by the 
Sparre Andersen theorem.
Moreover, 
the Sparre Andersen theorem can be derived 
on the basis of formula \eqref{WH} 
under the assumption that $\phi_{n+1}(\cdot)$ 
is of bounded variation \cite{frisch_etal-1975}.
Actually, the Sparre Andersen theorem  
\cite{sparre-andersen-1954b} is a fundamental result
in the study of the first-passage time problems.

Originally established in 1954 
\cite{sparre-andersen-1954b}, the Sparre Andersen theorem states
\be
\phi_n(0) 
= 2^{-2n} \binom{2n}{n}
\sim \frac{1}{\sqrt{n \pi}} \,, \quad n \to \infty \,.
\label{SA}
\ee
Thereby, the Sparre Andersen formula \eqref{SA} 
is valid for arbitrary but symmetric
jump-size distribution, i.e., $k(\xi)=k(-\xi)$ with $\xi\in\RR$.
In other words,
the Sparre Andersen theorem \cite{sparre-andersen-1954b}
provides the exact survival probability 
\ck{conditioned on the initial position}
for symmetric Markovian random walkers starting in the same location
of the absorbing barrier and states that
it is independent of the jump-size distribution $k(x)$
(whenever it is assumed to be symmetric). 
See, for example, F. Spitzer \cite{spitzer-1956} and 
W. Feller \cite[Section XII.7]{feller-1971}.
This independency uncovers a universal nature of the Sparre Andersen theorem
that is reflected also in other results, see, e.g., 
\cite{chechkin_etal-jpa-2003,koren_etal-pa-2007,koren_etal-prl-2007,
majumdar-pa-2010,dybiec_etal-jpa-2016,padash_etal-jpa-2019,
lacroix_etal-jpa-2020}.
Thus, we refer to this independency as
the {\it universality} of the Sparre Andersen theorem. 
Sparre Andersen formula \eqref{SA} is indeed one of the many results 
emerging from a deep {\it corpus studiorum} on random walks 
based on combinatorial arguments 
by E. Sparre Andersen 
\cite{sparre-andersen-1949,
sparre-andersen-1953, 
sparre-andersen-1954a, 
sparre-andersen-1954b} 
and F. Spitzer \cite{spitzer-1956,spitzer-1957,spitzer-1960}.
For more recent and general results
on the basis of probabilistic arguments, see, e.g., 
reference \cite{aurzada_etal-2015}. 
The general formulation of the first-passage time problem
for processes with discrete-time is available 
\gp{in the literature}
\cite{bray_etal-ap-2013,majumdar_etal-jpa-2017}.

\subsection{From the discrete- to the continuous-time setting}
\label{sec:DtoC}
\gp{The survival probability problem can be formulated also in
the continuous-time setting.}
\gp{In particular, we consider the CTRW
on homogeneous and continuous space in the uncoupled formulation
as originally introduced by Montroll \& Weiss in the year 1965 
\cite{montroll_etal-jmp-1965}.}
\ck{This means that  
the i.i.d. random waiting-times and the i.i.d. random jump-sizes are, 
at any epoch, independent of each other and also 
of the current position and time.}
\gp{Therefore, in analogy with the discrete-time setting,
the walker's position after $n$ iterations,
with $n\in \NN$, is given by the sum of $n$ 
i.i.d. random variables $X_n \in \RR$ distributed according to the
jump-size distribution $k(x)$, i.e., $S_n=x_0 + X_1 + \dots + X_n$
where $x_0 \in \RR_0^+$ is the initial position.
The actual time $t \ge 0$ is given
by the sum of $n$ i.i.d. random waiting-times $\tau_j \in \RR_0^+$ 
between two consecutive jumps, i.e.,
$t=\sum_{j=1}^n \tau_j$ with initial instant $t=0$.}
\gp{Moreover, let the survival probability 
conditioned on the initial position in the continuous-time setting 
be denoted by $\Lambda:\RR_0^+\times\RR_0^+\to[0,1]$.
Thus, by using the same approach adopted to compute 
the walker's density function in the CTRW theory
\cite{montroll_etal-jmp-1965}, 
the conditional survival probability $\Lambda$
is given by the weighted superposition of all
the possible discrete-time counterparts $\phi_n(\cdot)$ 
and it results in the series}
\be
\Lambda(x_0,t)=\sum_{n=0}^\infty
\phi_n(x_0) \, \Psi_n(t) \,, \quad x_0,t\in\RR_0^+\,,
\label{discrete_continuous_Lambda}
\ee
where the index $n$ counts the number of occurred jumps
and $\Psi_n(\cdot)$ 
is the probability to have an elapsed time 
equals to $t\in\RR_0^+$ after $n\in\NN_0$ jumps such that
\be
\Psi_n(t) = \int_0^t \Psi_{n-1}(t-\tau) \psi(\tau) \, d\tau \,,\quad
\Psi_0(t)=\Psi(t)=1-\int_0^t \psi(\tau) d\tau \,,
\label{def:Psi}
\ee
\gp{and} $\psi : \RR_0^+ \to \RR_0^+$
is the distribution of the waiting-times.

For the present purposes, 
we pass to the Laplace domain 
and we obtain for formula \eqref{discrete_continuous_Lambda}
\be
\mathcal{L}\{ \Lambda(x_0,t);s \}=\lLambda(x_0,s)
=\lPsi(s) \sum_{n=0}^\infty
\phi_n(x_0) \, \left[\lpsi(s)\right]^n \,,
\label{lambdactrw}
\ee
where $\displaystyle{
\mathcal{L}\{g(t);s\}:=\widetilde{g}(s)=\int_0^{+\infty}\,\e^{-st}\,g(t)\,dt}$
is the Laplace transform of a sufficiently well-behaving function $g(t)$.
Hence, after splitting formula \eqref{lambdactrw} in
\be
\lLambda(x_0,s)=
\lPsi(s) \phi_0(x_0) + 
\lPsi(s) \sum_{n=1}^\infty
\phi_n(x_0) \, \left[\lpsi(s)\right]^n\,,
\label{eq:intermediate}
\ee
by recalling that $\phi_0(x_0)=1\,,\forall\,x_0\in\RR_0^+\,,$ and 
by expressing $\phi_n(\cdot)$ through the 
Wiener--Hopf integral equation \eqref{WH},
from formula \eqref{eq:intermediate} it results that
\be
\lLambda(x_0,s)=\lPsi(s) + \lpsi(s)
\int_0^{+\infty} k(\xi-x_0) \lLambda(\xi,s) \, d\xi \,.
\label{WHctrw}
\ee
Finally, after the inverse Laplace 
transformation $g(t)=\mathcal{L}^{-1}\{\widetilde{g}(s);t\}$, 
the equation for the survival probability is
\be
\Lambda(x_0,t) = \Psi(t) + 
\int_0^t
\psi(t-\zeta) 
\int_{0}^{+\infty}
k(\xi-x_0) \Lambda(\xi,\zeta) \, d\xi d\zeta \,, \quad x_0,t\in\RR_0^+\,,
\label{eq:Lambda}
\ee
that is the analogue with continuous-time of formula \eqref{WH},
see also \ck{references} 
\cite{dahlenburg_etal-jpa-2022,dahlenburg_etal-prsa-2023,shkilev-pre-2024}, 
and we recover $\Lambda(x_0,0)=1\,,\forall\,x_0\in\RR_0^+\,,$ because
by definition \eqref{def:Psi} it holds $\Psi(0)=1$.

\subsection{The Sparre Andersen theorem in the framework of the CTRW}
The Sparre Andersen theorem can also be put in the framework
of the uncoupled CTRW,
in particular, through formula \eqref{discrete_continuous_Lambda}.
In fact, \gp{if we set $x_0=0$} in formula \eqref{lambdactrw}
and from the Sparre Anderson formula \eqref{SA} for the discrete-time setting
we take the expression of $\phi_n(0)$, 
then 
\begin{eqnarray}
\widetilde{\Lambda}(0,s)
&=& \widetilde\Psi(s) \, \sum^{\infty}_{n=0} 2^{-2n} \binom{2n}{n} \,
\left[ \widetilde\psi(s) \right]^n 
\nonumber\\
&=& \frac{\widetilde\Psi(s)}{\sqrt{1-\widetilde{\psi}(s)}} 
= \frac{\sqrt{s \, \widetilde\Psi(s)}}{s} \nonumber \\
&=&\frac{1}{s} \sqrt{1-\widetilde{\psi}(s)} \,,
\label{SA_non_Mark2}
\end{eqnarray}
where in the second line we used the series
\be
\sum_{n=0}^\infty 2^{-2n} \binom{2n}{n} z^n = \frac{1}{\sqrt{1-z}} \,,
\quad |z| < 1 \,.\label{ID_sum}
\ee
Formula \eqref{SA_non_Mark2} was already reported to some extent
\cite[formula (7)]{artuso_etal-pre-2014} but not formally derived, yet.
By applying the initial and final value theorems we have
\begin{subequations}
\be
\Lambda(0,0)=\lim_{s \to \infty} s \widetilde{\Lambda}(0,s) = 1 \,,
\ee
\be
\Lambda(0,\infty)=\lim_{s \to 0} s \widetilde{\Lambda}(0,s) = 0 \,,
\ee
\end{subequations}
because $\widetilde{\psi}(0)=1$ and $\widetilde{\psi}(\infty)=0$.
After the Laplace anti-transformation of the first line of formula
\eqref{SA_non_Mark2},
the Sparre Andersen theorem reads
\be
\Lambda(0,t) = \sum^{\infty}_ {n=0} 2^{-2n} \binom{2n}{n} \, \Psi_n(t) \,.
\label{SAnonM}
\ee

For Markovian random walks, 
we have that $\psi(t)=\psi_{\rm M}(t)=\e^{-t}$ 
\cite{zwanzig-jsp-1983,mainardi_etal-pa-2000}, such that 
\be
\lpsiM(s)=\lPsiM(s)=\frac{1}{1+s} \,.
\label{def:Markov}
\ee
and then formula \eqref{SA_non_Mark2} reads 
\be
\lLambdaM(0,s) = \frac{1}{\sqrt{s(1+s)}} \,.
\label{SActrw}
\ee
\ck{Consequently,} formula \eqref{SAnonM} for the Markovian case 
gives in the original domain 
\cite{shkilev-pre-2024}
\be
\LambdaM(0,t) = \e^{-t/2} \, I_0(t/2) \,, 
\label{SActrw2}
\ee
where $I_0(\cdot)$
is the modified Bessel function of the first kind of order $0$
defined by the series
\be
I_0(\zeta)=\sum_{j=0}^{\infty}\frac{(\zeta^2/4)^j}{(j!)^2} \,,
\quad \zeta \in \RR_0^+ \,,
\ee 
with
\be
I_0(0)=1 \,, \quad
\left. \frac{dI_0(\zeta)}{d\zeta}\right|_{\zeta=0}=0 \,, \quad
\left. \frac{d^2I_0(\zeta)}{d\zeta^2}\right|_{\zeta=0}=\frac{1}{2} \,.
\label{bessel0}
\ee
From 
formula \eqref{SActrw} we have in the long-time limit, i.e., $s \to 0$, 
that $\lLambdaM(0,s) \sim 1/\sqrt{s}$ and after the Laplace inversion it gives
$\LambdaM(0,t) \sim  1/\sqrt{t}$ for $t \to +\infty$. 
This last limit is 
the continuous-time counter-part of limit \eqref{SA}. 

For non-Markovian random walks, 
\gp{we consider the specific model} \cite{hilfer_etal-pre-1995}
\be
\psi(t)= t^{\beta-1} E_{\beta,\beta}(-t^\beta) \,, \quad 
\Psi(t)= E_{\beta,1}(-t^\beta)= E_\beta(-t^\beta) \,,
\label{def:nonMarkov2}
\ee
where $E_{\beta,\alpha}(z)$ is the Mittag--Leffler function 
\cite[Appendix E]{mainardi-2010}
\be
E_{\beta,\alpha}(z)=\sum_{j=0}^\infty \frac{z^j}{\Gamma(\beta j + \alpha)} \,,
\quad {\rm Re}\{\beta\} > 0 \,,
\quad \alpha \in \CC \,,
\quad z \in \CC \,,
\ee
that in the Laplace domain correspond to
\be
\lpsi(s)= \frac{1}{1+s^\beta} \,, \quad 
\lPsi(s)=\frac{s^{\beta-1}}{1+s^\beta} \,.
\label{def:nonMarkov}
\ee
\gp{Hence,} formula \eqref{SA_non_Mark2} reads 
\be
\lLambda(0,s) = \frac{s^{\beta/2-1}}{\sqrt{1+s^\beta}} \,.
\label{SActrwNM}
\ee
We observe that it holds
\be
\lLambda(0,s) = s^{\beta-1} \lLambdaM(0,s^\beta) \,, 
\label{LLM}
\ee
and by applying the Efros formula \cite{efros-1935,wlodarski-1952},
i.e., 
\be
\mathcal{L}^{-1}\!\left\{v(s)\widetilde{\mathcal{W}}[q(s)]\right\}=
\int_0^\infty \mathcal{W}(\zeta) \, 
\mathcal{L}^{-1}\!\left\{v(s)\e^{-\zeta q(s)}\right\} \, d\zeta \,,
\label{efros}
\ee
with $v(s)=s^{\beta-1}$ and $q(s)=s^\beta$,
in the original domain we have 
\be
\Lambda(0,t) = t^{-\beta} \int_0^\infty
\e^{-\zeta/2} I_0(\zeta/2) 
M_\beta(\zeta/t^\beta) 
\, d\zeta \,, 
\label{SActrw2NM}
\ee
where $M_\beta(\cdot)$ is the Mainardi/Wright function 
defined by the series
\cite[Appendix F]{mainardi-2010}
\be
M_\nu(z)=\sum_{j=0}^\infty 
\frac{(-1)^j}{j!} \frac{z^j}{\Gamma[-\nu j + (1-\nu)]} \,,
\quad 0 < \nu < 1 \,, \quad z \in \CC \,,
\label{def:M}
\ee
whose Laplace transform is known in \ck{the} literature 
\cite[(4.26)]{mainardi_etal-fcaa-2001}.
From formula \eqref{SActrwNM} we have 
that $\lLambda(0,s) \sim s^{\beta/2-1}$, when $s \to 0$, 
then, after the Laplace inversion it holds
$\Lambda(0,t) \sim  1/t^{\beta/2}$ for $t \to +\infty$. 

\gp{Moreover, after reminding that the  
unconditional survival probability is determined by the integral
$\displaystyle{\int_0^\infty \rho(\xi)\Lambda(\xi,t) \, d\xi}$,
where $\rho:\RR_0^+ \to \RR_0^+$ is the
destribution of the initial position,
then we can state
\begin{thm}
\label{thm:1}
The unconditional survival probability of a symmetric CTRW is
independent of the jump-size distribution if this
is equal to the distribution of the initial position.
\end{thm}
}
\begin{proof}
From the third line of formula \eqref{SA_non_Mark2} we have 
$\lPsi(s)=s \lLambda^2(0,s)$ and,
after summing and subtracting $\lLambda(0,s)$, we can derive
\be
\lLambda(0,s)=\lPsi(s)+\lLambda(0,s)\left[1-s\lLambda(0,s)\right] \,,
\ee
that, after the comparison against formula \eqref{WHctrw} with $x_0=0$, 
gives
\begin{eqnarray}
\int_0^\infty k(\xi) \lLambda(\xi,s) \, d\xi 
&=&\frac{\lLambda(0,s)}{\lpsi(s)}\left[1-s\lLambda(0,s)\right] \nonumber \\
&=&\frac{\lPsi(s)}{\lpsi(s)}\left[
\frac{1}{\sqrt{1-\lpsi(s)}} - 1 \right] \nonumber \\
&=&\frac{\sqrt{1-\lpsi(s)}-1+\lpsi(s)}{s \,\lpsi(s)} \,,
\label{SAintegral}
\end{eqnarray}
that is independent of the jump-size distribution $k(x)$
as a consequence of the universality of the 
Sparre Andersen theorem \eqref{SA} when plugged into 
\eqref{discrete_continuous_Lambda}. 
\end{proof}

\section{First-return time for CTRW}
\label{sec:FRT}
\subsection{Problem formulation and definitions}
The duration of excursions of a random walker 
\gp{until its first comeback to the starting position} is called
first-return time (FRT).
This observable can be understood with an example from
the animal kingdom in terms of the 
return of an animal to a previously occupied area
as the duration of the flight of birds 
when they leave and then return to their nests. 
This is a wide-spread behaviour associated
with a number of ecological processes as site fidelity, 
breeding, social associations, optimal foraging \cite{greenwood-ab-1980}
and successfully modelled also by fractional kinetics 
\cite{giuggioli_etal-jmb-2012,zeng_etal-fcaa-2014}.

\gp{The FRT for the uncouple CTRW reported at the beginning
of Subsection \S \ref{sec:DtoC}
can be exactly calculated by using integral formula \eqref{SAintegral}.
In particular, we can exactly calculate
the FRT for the fractional kinetics emerging
from a CTRW with Mittag--Leffler distributed waiting-times 
\eqref{def:nonMarkov2}.
In this Section, 
we consider a random walk starting at the origin $x=0$ and,}
provided that the first jump away is long $\xi$,
we are interested in how much time it takes
to the random walker to pass through the origin
for the first time after the departure.
As a matter of fact, 
\gp{the problem is equivalent to a first-passage time problem
that starts at the initial instant $\tau_0$
from the random first-landing site $\xi$} 
and the associated absorbing barrier is located in the origin $x=0$. 
Actually, \ck{the first-landing} site $\xi$ is 
distributed according to the jump-size distribution $k(\xi)$
and $\tau_0$ is the delay elapsed between the two locations $x=0$ and $x=\xi$. 
The CTRW formalism allows for two different formulations
of the problem: 
{\it first jump then wait} (jw), such that $\tau_0=0$, and
{\it first wait then jump} (wj), 
such that $\tau_0$ is a random variable
distributed according to $\psi(t)$.
Thus, by \gp{following the literature} 
\cite{ding_etal-pre-1995,colaiori_etal-pre-2004,jeon_etal-book-2014},
we define the FRT as follows:
\begin{equation}
\text{FRT} \equiv
\begin{cases}
\mathcal{T}_{\ell}(\xi) + \tau_0 \,; & \xi\in\mathbb{R}^- \,,\\
& \\
\mathcal{T}_{r}(\xi) + \tau_0 \,; &\xi\in\mathbb{R}^+ \,,
\end{cases}
\label{eq2}
\end{equation}
where $\mathcal{T}$ is the first-passage time and,
in particular, $\mathcal{T}_{\ell}$ is the first-passage time 
when the first jump away from the origin
is towards the left into the position $\xi <0$ and  
$\mathcal{T}_{r}$ when it is towards the right into the position $\xi >0$. 

Let $\displaystyle{\lambda(\xi,t)=-\frac{\partial \Lambda}{\partial t}}$ 
be the first-passage time density 
conditioned to the starting position $\xi$ and 
satisfying the normalisation condition
$\displaystyle{\int_0^\infty \lambda(\xi,t) \, dt = 1}$,
then the density function of the first-passage time
weighted over all the possibile first-jump landing positions results to be 
\begin{eqnarray}
f(t)
&=&\int_{-\infty}^{+\infty} k(\xi)\lambda(\xi,t) \, d\xi \nonumber \\
&=& 2 \int_0^{+\infty} k(\xi)\lambda(\xi,t) \, d\xi \,,
\label{eq5}
\end{eqnarray}
where the symmetry property of the jump-size distribution 
$k(\xi)=k(-\xi)$ is used.

Thus, since $\mathcal{T}$ and $\tau_0$ are statistically independent,
we have that the density function of their sum,
namely the FRT \eqref{eq2}, is
\begin{subequations}
\be
\mathcal{P}^{\rm jw}(t)=f(t) \,, 
\quad {\rm when} \quad \tau_0=0 \,,
\label{Pjw}
\ee
\be
\mathcal{P}^{\rm wj}(t)
=\int_0^t \psi(t-\zeta) \, \mathcal{P}^{\rm jw}(\zeta) \, d\zeta \,,
\quad {\rm when} \quad \tau_0 \sim \psi(t) \,.
\label{Pwj}
\ee
\end{subequations}

An alternative formulation is possible by starting from the
discrete-time setting with the extension to the continuous-time setting 
in the same fashion of formula \eqref{discrete_continuous_Lambda}
\cite{artuso_etal-pre-2014}. 
In particular, by defining the discrete-time 
first-passage time density 
$\lambda_n(\xi) = \phi_{n-1}(\xi)-\phi_n(\xi)$,
with $n \ge 1$, and exploiting the Sparre Andersen theorem \eqref{SA}
at $\xi=0$, we have \cite{redner-2001}
\begin{eqnarray}
\lambda_n(0)
&=& 2^{-2(n-1)}\binom{2(n-1)}{n-1}-2^{-2n}\binom{2n}{n} \nonumber \\
&=& 2^{-2(n-1)}\binom{2n}{n} \left\{ \frac{n}{4n-2}-\frac{1}{4} \right\}
\nonumber \\
&=& 2^{-2(n-1)}\binom{2n}{n} \frac{1}{2(4n-2)} \nonumber \\
&=&\frac{2^{-2n}}{2n-1}\binom{2n}{n} \,,
\label{eq7}
\end{eqnarray}
that correctly gives $\sum_{n=1}^\infty\,\lambda_n(0)=1$.
Moreover, by following definition \eqref{eq5},
in the discrete-time setting we have 
\begin{eqnarray}
f_n
&=& 2 \int_0^\infty k(\xi) \lambda_n(\xi) \, d\xi \nonumber \\
&=& 2 \int_0^\infty k(\xi) \left[
\phi_{n-1}(\xi) - \phi_n(\xi) \right] \, d\xi \nonumber \\
&=& 2 \left[
\phi_{n}(0) - \phi_{n+1}(0) \right] = 2 \lambda_{n+1}(0) \,,
\quad n \ge 1 \,,
\label{discrete}
\end{eqnarray}
where in the third line we used \eqref{WH}, and it holds
$f_0=0$.

\subsection{The (jw)-case}
We remind that in this case $\mathcal{P}^{\rm jw}(t)=f(t)$
\eqref{Pjw}, thus in the following we study the density function $f(t)$.

In particular,
thanks to the relation 
in the Laplace domain between the 
first-passage time density and the corresponding survival probability,
i.e., \gp{$\widetilde{\lambda}(\xi,s)=1-s \, \widetilde{\Lambda}(\xi,s)$}, 
from \eqref{eq5} we have
\begin{eqnarray}
\widetilde{f}(s)
&=& 2 \int_0^{+\infty} k(\xi) \widetilde{\lambda}(\xi,s) \, d\xi
\nonumber \\
&=&1-2 s \int_{0}^\infty k(\xi)\widetilde{\Lambda}(\xi,s) \, d\xi 
\nonumber \\
&=&1 - 2 \, \frac{\sqrt{1-\widetilde{\psi}(s)}-1+\widetilde{\psi}(s)}
{\widetilde{\psi}(s)} \nonumber \\
&=&\frac{2-\widetilde{\psi}(s)-2\sqrt{1-\widetilde{\psi}(s)}}
{\widetilde{\psi}(s)} \,,
\label{eq6}
\end{eqnarray}
where formula \eqref{SAintegral} has been used in 
the third line. 
It can be checked that the normalisation condition $\widetilde{f}(0)=1$
is met because $\widetilde{\psi}(0)=1$.
\gp{We underline that,
since formula \eqref{SAintegral} from Theorem \ref{thm:1} has been used
for deriving \eqref{eq6}, 
we have the following
\begin{remark}
\label{rmk:1}
The density function of the FRT of a symmetric CTRW in the (jw)-case
is independent of the jump-size distribution.
\end{remark}}
Moreover, if we consider the non-Markovian setting \eqref{def:nonMarkov},
from formula \eqref{eq6}
we finally obtain that \ck{the Laplace transform of} 
the density function of the FRT for a CTRW 
(jw)-type is
\be
\widetilde{f}(s)=1+2s^\beta - 2 \sqrt{s^\beta}\sqrt{s^\beta+1} \,.
\label{ftilde}
\ee
When $\beta=1$ the system is Markovian and therefore, 
by denoting with $f_{\rm M}(t)$ the 
corresponding density of the FRT, 
the following equality in the Laplace domain holds
\be
\widetilde{f}(s)=\widetilde{f}_{\rm M}(s^\beta) \,.
\label{ftildefM}
\ee 
\gp{Hence, we can state
\begin{thm}
Let $f(t)$ and $f_{\rm M}(t)$ be the density functions of the FRT
in the $(jw)$-case for a symmetric CTRW 
with Mittag--Leffler \eqref{def:nonMarkov2} and 
exponentially distributed waiting-times, respectively, 
then it holds
\be
f(t)=\int_0^\infty 
\zeta^{-1/\beta} \ell_\beta(t/\zeta^{1/\beta}) f_{\rm M}(\zeta) \, d\zeta \,,
\quad 0 < \beta < 1 \,,
\label{f}
\ee
where $\ell_\beta(t)$ is the one-sided L\'evy stable density
characterised by the Laplace transform $\displaystyle{\e^{-s^\beta}}$.
\end{thm}
}
\begin{proof}
By reminding the Efros theorem \eqref{efros}, 
formula \eqref{ftildefM} can be inverted
by setting $v(s)=1$ and $q(s)=s^{\beta}$ such that we have
the integral relation \eqref{f}.
\end{proof}
We observe that from \eqref{f} the Markovian case is recovered
when $\beta=1$ because $\ell_1(t)=\delta(t-1)$.

In the Markovian case with $\beta=1$, 
from \eqref{ftilde} we have
\begin{eqnarray}
\widetilde{f}_{\rm M}(s) 
&=& (1+s) \, \frac{1+2s-2\sqrt{s}\sqrt{s+1}}{1+s} \nonumber \\
&=& \frac{1}{1+s} + 2s \left[\frac{1}{1+s} - \frac{1}{\sqrt{s(1+s)}}\right] 
\nonumber \\
& & \qquad \qquad + s \, \left\{
\frac{1}{1+s} + 2s \left[\frac{1}{1+s} - \frac{1}{\sqrt{s(1+s)}}\right] \right\} 
\,.
\end{eqnarray}
\ck{Thus,} by remembering the properties of the Laplace transform,
and also the Laplace transform pair \eqref{SActrw}-\eqref{SActrw2},
we can calculate the exact result
\begin{eqnarray}
f_{\rm M}(t)
&=& \e^{-t/2}\left[
I_0(\zeta) - \frac{dI_0(\zeta)}{d\zeta}\right]_{\zeta=t/2} -\e^{-t} 
\nonumber \\
& & \qquad \qquad + \frac{d}{dt}\left\{
\e^{-t/2}\left[
I_0(\zeta) - \frac{dI_0(\zeta)}{d\zeta}\right]_{\zeta=t/2} -\e^{-t} 
\right\} \nonumber \\
&=& \frac{1}{2} \e^{-t/2}\left[
I_0(\zeta) - \frac{d^2I_0(\zeta)}{d\zeta^2}\right]_{\zeta=t/2}
\,,
\label{fexplicit}
\end{eqnarray} 
and then \gp{the integral} formula \eqref{f} is fully determined.

From \eqref{fexplicit}, by \gp{using the values given}
in \eqref{bessel0}, we have that
\be
f_{\rm M}(0)=\frac{1}{4} \,,
\ee
and, when plugged in \eqref{eq5}, 
it gives the following universal result independent of $k(\xi)$
\be
f_{\rm M}(0)
=\int_{-\infty}^{+\infty} k(\xi)\lambda_{\rm M}(\xi,0) \, d\xi=\frac{1}{4} \,,
\ee 
in opposition to the first-passage time density that 
is indeed dependent on $k(\xi)$ also at $t=0$ according to 
$\displaystyle{\lambda_{\rm M}(\xi,0)=\int_\xi^\infty k(y)dy}$.
In the non-Markovian case it holds
\be
f(0)=+\infty \,,
\ee
in analogy with the first-passage time density $\lambda(\xi,0)$.
In fact,
after the change of variable $\zeta=(t/\chi)^\beta$ in \eqref{f}, we have
\begin{eqnarray}
f(t)
&=& \beta t^{\beta - 1} \int_0^\infty \ell_\beta(\chi)f_{\rm M}(t^\beta/\chi^\beta)
\, \frac{d\chi}{\chi^\beta} \nonumber \\
&\sim& C \, t^{\beta-1} \,, \quad t \to 0 \,, \quad
{\rm with} \quad C=\frac{\beta}{4} \int_0^\infty \ell_\beta(\chi)\, \frac{d\chi}{\chi^\beta} 
\,.
\end{eqnarray}

From \eqref{ftilde} we can derive the asymptotic behaviour for
large elapsed times.
\ck{Actually,} we have $\widetilde{f}(s) \sim 1 - 2s^{\beta/2}$ when $s \to 0$
that is consistent with the behaviour of 
\gp{$\displaystyle{\e^{-2s^{\beta/2}}}$} for small $s$ and gives
\be
f(t) \sim \frac{1}{2^{2/\beta}} \,
\ell_{\beta/2}(t/2^{2/\beta}) \sim \frac{2}{t^{\beta/2+1}} \,,
\quad t \to +\infty \,.
\label{fasymptotic}
\ee
\ck{From \eqref{fasymptotic} it} follows that the mean FRT in the (jw)-case
is infinite both in Markovian and non-Markovian setting 
but the right tail of the density function decreases
slower in non-Markovian systems.
Because of formula \eqref{eq5},
we have that scaling \eqref{fasymptotic} is also the scaling 
of the corresponding first-passage time density 
$\lambda(\xi,t)$ \cite{jeon_etal-book-2014}, 
notwithstanding the exact functions may differ.

Tools of fractional calculus can be used for studying the 
integral formula \eqref{f}. In fact,
\gp{let} $J^\mu$, with $\mu > 0$, be the Riemann--Liouville fractional integral
defined by the Laplace symbol $s^{-\mu}$ \cite{gorenflo_etal-cism-1997},
then in the Laplace domain it holds
\be
\mathcal{L}\{J^{1-\beta}f(t);s\}=
\frac{\widetilde{f}(s)}{s^{1-\beta}}=
\int_0^\infty \frac{\e^{-\zeta s^\beta}}{s^{1-\beta}} 
f_{\rm M}(\zeta) \, d\zeta \,,
\ee
that, by using the formula 
$\displaystyle{\mathcal{L}\{t^{-\beta} M_\beta(\zeta t^{-\beta});s\}
=s^{\beta-1}\e^{-\zeta s^\beta}}$, with ${\rm Re}\{s\}>0$, 
\cite[(4.26)]{mainardi_etal-fcaa-2001}, gives
\be
J^{1-\beta}f(t)=\int_0^\infty 
\frac{1}{t^\beta}M_\beta(\zeta/t^{\beta}) f_{\rm M}(\zeta) \, d\zeta
\,. 
\label{Jf}
\ee
\ck{Since} from the study of time-fractional diffusion equations 
\cite{mainardi_etal-fcaa-2001}
we have that \linebreak[4] $\lim_{t \to 0} t^{-\beta}M_\beta(\zeta/t^{\beta}) = \delta(\zeta)$,
then \ck{it holds}
\be
\left. J^{1-\beta}f(t) \right|_{t=0}= f_{\rm M}(0) = \frac{1}{4} \,.
\label{Jf0}
\ee
Furthermore, 
we can also have the Laplace transform of the cumulative density function
$\displaystyle{F^{\rm jw}(t)=\int_0^t f(\zeta)d\zeta}$ that,
by using \eqref{ftildefM}, leads to the equality 
\be
\widetilde{F}^{\rm jw}(s)
=\frac{\widetilde{f}(s)}{s}
=s^{\beta-1} \frac{\widetilde{f}_{\rm M}(s^\beta)}{s^\beta}
=s^{\beta-1} \widetilde{F}^{\rm jw}_{\rm M}(s^\beta) \,.
\ee
\gp{Therefore, we can state the following
\begin{thm}
\label{thm:CDF}
Let $F^{\rm jw}(t)$ and $F^{\rm jw}_{\rm M}(t)$ 
be the cumulative density functions of the FRT in the $(jw)$-case
for a symmetric CTRW  
with Mittag--Leffler \eqref{def:nonMarkov2} and 
exponentially distributed waiting-times, respectively, 
then it holds
\be
F^{\rm jw}(t)=t^{-\beta}\int_0^\infty 
M_\beta(\zeta/t^\beta) F^{\rm jw}_{\rm M}(\zeta) \, d\zeta \,,
\quad 0 < \beta < 1 \,,
\label{CDF}
\ee
where $M_\beta(\zeta)$ is the Mainardi/Wright function
defined in \eqref{def:M}.
\end{thm}
}
\begin{proof}
By applying Efros formula \eqref{efros} in analogy with 
the pair \eqref{LLM} and \eqref{SActrw2NM},
we have \eqref{CDF}.
\end{proof}
To conclude, we remind the formula 
\cite[(6.3)]{mainardi_etal-fcaa-2003}
\be
t^{-\beta}M_\beta(\zeta/t^{\beta})=
t^{-\nu} \int_0^\infty
M_\eta(\zeta/y^{\eta})
M_\nu(y/t^{\nu}) \, \frac{dy}{y^\eta} \,,
\quad \beta=\nu \eta \,,
\label{MMM}
\ee
and then
\gp{from Theorem \ref{thm:CDF} we have the following
\begin{cor}
\label{cor:CDF2}
The relation between the cumulative density function
with anomalous parameter $\beta$ and that of parameter 
$\eta > \beta$ is:
\be
F^{\rm jw}(t;\beta)=t^{-\nu}\int_0^\infty 
M_\nu(\zeta/t^\nu) F^{\rm jw}(\zeta;\eta) \, d\zeta \,,
\quad \beta=\nu \eta \,.
\label{CDF2}
\ee
\end{cor}
\begin{proof}
By using \eqref{MMM} in \eqref{CDF} we have \eqref{CDF2}.
\end{proof}
\begin{remark}
An interesting special case of formula \eqref{CDF2} is obtained when $\nu=1/2$,
that is
\be
F^{\rm jw}(t;\beta/2)=2\int_0^\infty 
\frac{\e^{-\zeta^2/(4t)}}{\sqrt{4 \pi t}} F^{\rm jw}(\zeta;\beta) \, d\zeta \,,
\label{CDF3}
\ee
where the identity $M_{1/2}(\zeta)=\e^{-\zeta^2/4}/\sqrt{\pi}$ has 
been used.
\end{remark}
}

A large number of other formulae can be obtained by making use
of the existing results for the Mainardi/Wright function, e.g., 
\cite{mainardi_etal-fcaa-2001,mainardi-2010,
mainardi_etal-ijde-2010,pagnini-fcaa-2013}.

\subsection{The (wj)-case}
The FRT density function in the (wj)-case 
\eqref{Pwj} can be studied in the Laplace domain.
In particular,
by using \eqref{eq5} and \eqref{eq6}, from \eqref{Pwj} we have
\begin{eqnarray}
\widetilde{\mathcal{P}}^{\rm wj}(s)
&=&\widetilde{\psi}(s) \widetilde{f}(s)
\nonumber \\
&=&2-\widetilde{\psi}(s)-2\sqrt{1-\widetilde{\psi}(s)} \,.
\label{eq6wj}
\end{eqnarray}
It can be checked that the normalisation condition 
$\widetilde{\mathcal{P}}^{\rm wj}(0)=1$ is met and 
by applying the initial and the final value theorems it holds
\begin{subequations}
\be
\mathcal{P}^{\rm wj}(0)=\lim_{s \to +\infty} s \widetilde{\mathcal{P}}^{\rm wj}(s) = 0 \,,
\label{initialwj}
\ee
\be
\mathcal{P}^{\rm wj}(+\infty)=\lim_{s \to 0} s \widetilde{\mathcal{P}}^{\rm wj}(s) = 0 \,,
\label{finalwj}
\ee
\end{subequations}
because $\widetilde{\psi}(0)=1$ and $\widetilde{\psi}(+\infty)=0$.
\gp{In analogy with Remark \ref{rmk:1}, we underline 
also in this case that,
since Theorem \ref{thm:1} lays behind the derivation of \eqref{eq6wj}, 
we have the following
\begin{remark}
\label{rmk:2}
The density function of the FRT of a symmetric CTRW in the (wj)-case
is independent of the jump-size distribution.
\end{remark}}
Formulae \eqref{initialwj} and \eqref{finalwj} hold both 
in the Markovian and non-Markovian case.
More explicitly, from \eqref{def:nonMarkov} we have 
that the density function of the FRT for a CTRW 
(wj)-type is
\be
\widetilde{\mathcal{P}}^{\rm wj}(s)
=\frac{1+2s^\beta - 2 \sqrt{s^\beta}\sqrt{s^\beta+1}}{s^\beta + 1} \,.
\label{ftildewj}
\ee
When $\beta=1$ the system is Markovian and therefore, 
by denoting with $\mathcal{P}_{\rm M}^{\rm wj}(t)$ the 
corresponding density of the FRT, in analogy with \eqref{ftildefM},
in the Laplace domain it holds
\be
\widetilde{\mathcal{P}}^{\rm wj}(s)
=\widetilde{\mathcal{P}}_{\rm M}^{\rm wj}(s^\beta) \,,
\label{ftildefMwj}
\ee 
such that, by applying the Efros formula \eqref{efros},
it can be inverted and we have 
\gp{
\begin{thm}
Let $P^{\rm wj}(t)$ and $P^{\rm wj}_{\rm M}(t)$ 
be the density functions of the FRT in the $(wj)$-case
for a symmetric CTRW  
with Mittag--Leffler \eqref{def:nonMarkov2} and 
exponentially distributed waiting-times, respectively, 
then it holds
\be
\mathcal{P}^{\rm wj}(t) =\int_0^\infty 
\zeta^{-1/\beta} \ell_\beta(t/\zeta^{1/\beta}) 
\, \mathcal{P}_{\rm M}^{\rm wj}(\zeta) \, d\zeta \,,
\quad 0 < \beta < 1 \,,
\label{Pwj2}
\ee
which recasts formula \eqref{f} for the (jw)-case.
\end{thm}
}

In the Markovian case with $\beta=1$, 
from \eqref{ftildewj} we have
\be
\widetilde{\mathcal{P}}^{\rm wj}_{\rm M}(s) =
\frac{1}{1+s} + 2s \left[\frac{1}{1+s} - \frac{1}{\sqrt{s(1+s)}}\right] \,, 
\ee
and, by remembering the properties of the Laplace transform
\gp{together with} the Laplace transform pair \eqref{SActrw}-\eqref{SActrw2},
we can calculate the exact result
\be
\mathcal{P}_{\rm M}^{\rm wj}(t)=
\e^{-t/2}\left[
I_0(\zeta) - \frac{dI_0(\zeta)}{d\zeta}\right]_{\zeta=t/2}
-\e^{-t} \,,
\quad \mathcal{P}_{\rm M}^{\rm wj}(0)=0 \,,
\label{PMexplictwj}
\ee 
\gp{such that} the integral formula \eqref{Pwj2} is fully determined.

We can quantify the difference between the two cases (jw) and (wj).
In fact, first by comparing \eqref{fexplicit} and 
\eqref{PMexplictwj} we have
\gp{
\begin{prop}
The relationship between the density functions of the FRT 
with exponentially distributed waiting-times
in the $(wj)$- and $(jw)$-case is
\be
\mathcal{P}_{\rm M}^{\rm wj}(t)=2\mathcal{P}_{\rm M}^{\rm jw}(t)
-
\e^{-t/2}\left[\frac{dI_0}{d\zeta} - \frac{d^2 I_0}{d\zeta^2}
\right]_{\zeta=t/2} - \e^{-t}
\,.
\ee 
\end{prop}
}
\ck{Later}, by using \eqref{f} and \eqref{Pwj2}, \ck{we find}
\begin{eqnarray}
\mathcal{P}^{\rm wj}(t)
&=&2\mathcal{P}^{\rm jw}(t)
- \int_0^\infty \zeta^{-1/\beta} \ell_\beta(t/\zeta^{1/\beta}) 
\, \e^{-\zeta} \, d\zeta \qquad \qquad \nonumber \\
& & \qquad 
- \int_0^\infty \zeta^{-1/\beta} \ell_\beta(t/\zeta^{1/\beta}) 
\, \e^{-\zeta/2}\left[\frac{dI_0}{d\chi} - \frac{d^2 I_0}{d\chi^2}
\right]_{\chi=\zeta/2} \, d\zeta \,. 
\end{eqnarray} 
The second term in the r.h.s. can be solved and put in a more clear form.
In fact, by using some formula concerning
L\'evy density function \cite[(4.26)]{mainardi_etal-fcaa-2001},
Mainardi/Wright function \cite[(4.32)]{mainardi_etal-fcaa-2001}
and Mittag--Leffler function \cite[(1.45)]{mainardi-2010},
we have the following equalities
\begin{eqnarray}
\int_0^\infty \zeta^{-1/\beta} \ell_\beta(t/\zeta^{1/\beta}) 
\, \e^{-\zeta s^\beta} \, d\zeta 
&=& 
\frac{\beta}{t}\int_0^\infty \frac{\zeta}{t^\beta} M_\beta(\zeta/t^\beta) 
\, \e^{-\zeta s^\beta} \, d\zeta \nonumber \\
&=& 
- \frac{s^{1-\beta}}{t}\frac{\partial}{\partial s}
\int_0^\infty M_\beta(y) \, \e^{-t^\beta s^\beta y} \, dy \nonumber \\
&=& 
- \frac{s^{1-\beta}}{t}\frac{\partial}{\partial s}
E_\beta(-t^\beta s^\beta) \nonumber \\
&=& t^{\beta-1} E_{\beta,\beta}(-t^\beta s^\beta) \,.
\end{eqnarray}
\ck{Thus,} by remembering \eqref{def:nonMarkov2} and by setting $s=1$, it holds
\be
\int_0^\infty \zeta^{-1/\beta} \ell_\beta(t/\zeta^{1/\beta}) 
\, \e^{-\zeta} \, d\zeta = \psi(t) \,,
\ee 
such that
\gp{
\begin{prop}
The relationship between the density functions of the FRT 
with Mittag--Leffler distributed waiting-times \eqref{def:nonMarkov2}
in the $(wj)$- and $(jw)$-case is
\begin{eqnarray}
\!\!
\mathcal{P}^{\rm wj}(t)
&=&2\mathcal{P}^{\rm jw}(t)
- \psi(t) \qquad \qquad \nonumber \\
& & \qquad 
- \int_0^\infty \zeta^{-1/\beta} \ell_\beta(t/\zeta^{1/\beta}) 
\, \e^{-\zeta/2}\left[\frac{dI_0}{d\chi} - \frac{d^2 I_0}{d\chi^2}
\right]_{\chi=\zeta/2} \, d\zeta \,. 
\end{eqnarray} 
\end{prop}
}

From \eqref{ftildewj} we can derive the asymptotic behaviour for
large elapsed times and, in analogy with \eqref{fasymptotic}, 
it is
\be
\mathcal{P}^{\rm wj}(t) \sim \frac{1}{2^{2/\beta}} \,
\ell_{\beta/2}(t/2^{2/\beta}) \sim \frac{2}{t^{\beta/2+1}} \,,
\quad t \to +\infty \,,
\label{fasymptoticwj}
\ee
from which it follows that the mean FRT also in the (wj)-case
is infinite both in Markovian
and non-Markovian setting and the right tail of the density function
decreases slower in non-Markovian systems.

The analogies between the (jw)- and the (wj)-case include
also formulae \eqref{Jf} and \eqref{Jf0} 
that now read as follows
\be
J^{1-\beta}\mathcal{P}^{\rm wj}(t)=\int_0^\infty 
\frac{1}{t^\beta}M_\beta(\zeta/t^{\beta}) 
\mathcal{P}^{\rm wj}_{\rm M}(\zeta) \, d\zeta
\,, 
\label{Jfwj}
\ee
\be
\left. J^{1-\beta}\mathcal{P}^{\rm wj}(t) \right|_{t=0}
= \mathcal{P}^{\rm wj}_{\rm M}(0)=0 \,,
\label{Jf0wj}
\ee
\gp{as well as the analogue of Theorem \ref{thm:CDF}, that is
\begin{thm}
\label{thm:CDFwj}
Let $F^{\rm wj}(t)$ and $F^{\rm wj}_{\rm M}(t)$ 
be the cumulative density functions of the FRT in the $(jw)$-case
for a symmetric CTRW  
with Mittag--Leffler \eqref{def:nonMarkov2} and 
exponentially distributed waiting-times, respectively, 
then it holds
\be
F^{\rm wj}(t)=t^{-\beta}\int_0^\infty 
M_\beta(\zeta/t^\beta) F^{\rm wj}_{\rm M}(\zeta) \, d\zeta \,,
\quad 0 < \beta < 1 \,.
\label{CDFwj}
\ee
where $M_\beta(\zeta)$ is the Mainardi/Wright function
defined in \eqref{def:M}.
\end{thm}
\begin{proof}
Same proof of Theorem \ref{thm:CDF}.
\end{proof}
Moreover, 
from Theorem \ref{thm:CDFwj} and by using \eqref{MMM},
we have
\begin{cor}
\label{cor:CDF4}
The relation between the cumulative density function
with anomalous parameter $\beta$ and that of parameter 
$\eta > \beta$ is:
$$
F^{\rm wj}(t;\beta)=t^{-\nu}\int_0^\infty 
M_\nu(\zeta/t^\nu) F^{\rm wj}(\zeta;\eta) \, d\zeta \,,
\quad \beta=\nu \eta \,,
$$
which is the analogue of formula \eqref{CDF2}.
\end{cor}
\begin{proof}
Same proof of Corollary \ref{cor:CDF2}.
\end{proof}
\begin{remark}
An interesting special case of Corollary \ref{cor:CDF4} 
is obtained when $\nu=1/2$, that is
$$
F^{\rm wj}(t;\beta/2)=2\int_0^\infty 
\frac{\e^{-\zeta^2/(4t)}}{\sqrt{4 \pi t}} F^{\rm wj}(\zeta;\beta) \, d\zeta \,,
$$
which is the analogue of formula \eqref{CDF3}.
\end{remark}
}

\begin{acknowledgement}[title={Acknowledgments}]
The authors are grateful to the organisers of the
{\it First Lucia Workshop on Fractional Calculus and Processes},
V{\"a}xj{\"o}, 9--13 December 2024, for their kind invitation and
for the pleasant time that inspired this research.
\end{acknowledgement}

\begin{funding}
This research is supported by the Basque Government through 
the BERC 2022--2025 program and
by the Ministry of Science and Innovation: BCAM Severo Ochoa accreditation CEX2021-001142-S / MICIN / AEI / 10.13039/501100011033.
\end{funding}

\bibliographystyle{vmsta2-mathphys}
\bibliography{all}

@string{fcaa = "Fract. Calc. Appl. Anal."}

@string{prl= "Phys. Rev. Lett."}

@string{pre= "Phys. Rev. E"}

@book{mainardi-2010,
        author = {F. Mainardi},
        title = {Fractional Calculus and Waves in Linear Viscoelasticity},
        publisher = {Imperial College Press},
        year = {2010},
        location = {London}
}

@article{mainardi_etal-ijde-2010,
author="F. Mainardi and A. Mura and G. Pagnini",
title="The {M-Wright} function in time-fractional diffusion processes: {A} tutorial survey",
journal="Int. J. Differ. Equations",
year="2010",
volume="2010",
number="",
pages="104505",
note=""
}

@article{klafter_etal-pw-2005,
author="J. Klafter and I. M. Sokolov", 
title="Anomalous diffusion spread its wings",
journal="Physics World",
year="2005",
volume="18",
number="",
pages="29--32",
note=""
}

@article{sokolov_etal-pt-2002,
author="I. M. Sokolov and J. Klafter and A. Blumen", 
title="Fractional kinetics",
journal="Physics Today",
year="2002",
volume="55",
number="",
pages="48--54",
note=""
}

@article{mainardi_etal-fcaa-2001,
     author="F. Mainardi and Yu. Luchko and G. Pagnini", 
     title="The fundamental solution of the space-time 
             fractional diffusion equation",
    journal=fcaa,
       year="2001",
     volume="4",
     number="2",
      pages="153--192",
      month="",
       note=""
}

@article{mainardi_etal-fcaa-2003,
     author="F. Mainardi and G. Pagnini and R. Gorenflo",
     title="Mellin transform and subordination laws in fractional diffusion processes",
    journal=fcaa,
       year="2003",
     volume="6",
     number="4",
      pages="441--459",
      month="",
       note=""
}

@article{pagnini-fcaa-2013,
     author="G. Pagnini",
     title="The {M-Wright} function as a generalization of the {Gaussian} density for fractional diffusion processes",
    journal=fcaa,
       year="2013",
     volume="16",
     number="2",
      pages="436--453",
      month="",
       note=""
}

@book{feller-1971,
     author="W. Feller",
     editor="",
      title="An Introduction to Probability Theory and its Applications", 
  publisher="Wiley",
       year="1971",
     volume="2",
     number="",
     series="",
    address="New York",
    edition="second",
      month="",
       note=""
}

@article{montroll_etal-jmp-1965,
	Author = "E. W. Montroll and G. H. Weiss",
	Journal = "J. Math. Phys.",
	Number = "2",
	Pages = "167--181",
	Title = "Random Walks on Lattices. {II}",
	Volume = "6",
	Year = "1965"
	}

@article{hilfer_etal-pre-1995,
	author = {R. Hilfer and L. Anton},
	title = {Fractional master equations and fractal time random walks},
	journal = {Phys. Rev. E}, 
	volume = {51}, 
	number = {2}, 
	year = {1995}, 
	pages = {R848--R851}
}

@article{mainardi_etal-pa-2000,
	author = {F. Mainardi and M. Raberto and R. Gorenflo and E. Scalas},
	title = {Fractional calculus and continuous-time finance {II}: the waiting-time distribution},
	journal = {Physica A}, 
	volume = {287}, 
	number = {3--4},
	year = {2000}, 
	pages = {468--481}
}

@incollection{gorenflo_etal-cism-1997,
     author="R. Gorenflo and F. Mainardi",
      title="Fractional calculus: 
integral and differential equations of fractional order",
  booktitle="Fractals and Fractional Calculus in Continuum Mechanics", 
     editor="A. Carpinteri and F. Mainardi",
  publisher="Springer--Verlag", 
address="Wien and New York", 
       year="1997",
     volume="",
     number="",
      pages="223--276",
      chapter="",
      month="",
       note=""
}

@article{chechkin_etal-jpa-2003, 
year="2003",
volume="36", 
pages="L537--L544",
journal="J. Phys. A: Math. Gen.",
title="First passage and arrival time densities 
for {L\'evy} flights and the failure of the method of images",
author="A. V. Chechkin and R. Metzler and V. Y. Gonchar
and J. Klafter and L. V. Tanatarov",
}

@incollection{chechkin_etal-anotrans-2008,
author="A. V. Chechkin and R. Metzler 
and J. Klafter and V. {Yu.} Gonchar",
title="Introduction to the Theory of {L\'evy} Flights",
editor="R. Klages and G. Radons and I. M. Sokolov",
booktitle="Anomalous Transport: Foundations and Applications",
publisher="Wiley--VCH Verlag GmbH \& Co. KGaA",
year="2008", 
address="Weinheim", 
     volume="",
     number="",
      pages="129--162",
      chapter="5",
      month="",
       note=""
}

@article{zeng_etal-fcaa-2014,
author="C. Zeng and {Y.Q.} Chen",
title="Optimal random search, fractional dynamics and fractional calculus",
journal=fcaa,
year="2014",
volume="17", 
number="",
pages="321--332",
note=""
}

@article{greenwood-ab-1980,
author="P. J. Greenwood",
title="Mating systems, philopatry and dispersal
in birds and mammals",
journal="Anim. Behav.",    
volume="28",
pages="1140--1162",
year="1980"
}

@article{giuggioli_etal-jmb-2012,
journal="J. Math. Biol.",
year="2012",
volume="64",
pages="647--656",
title="Linking animal movement to site fidelity",
author="L. Giuggioli and F. Bartumeus"
}

@article{majumdar_etal-jpa-2017,
journal="J. Phys. A: Math. Theor.",
volume="50",
year="2017",
pages="465002",
title="Survival probability of random walks and {L\'evy} flights
on a semi-infinite line",
author="S. N. Majumdar and P. Mounaix and G. Schehr"
}

@article{padash_etal-jpa-2019,
journal="J. Phys. A: Math. Theor.",
volume="52",
year="2019",
pages="454004",
title="First-passage properties of asymmetric {L\'evy flights}",
author="A. Padash and A. V. Chechkin and B. Dybiec and
I. Pavlyukevich and B. Shokri and R. Metzler",
}

@article{bray_etal-ap-2013,
author="A. J. Bray and S. N. Majumdar and G. Schehr",
year="2013",
title="Persistence and first-passage properties in nonequilibrium systems", 
journal="Adv. Phys.", 
volume="62",
issue="3", 
pages="225--361"
}

@book{estrada_kanwal-2000,
title="Singular Integral Equations",
author="R. Estrada and R. P. Kanwal",
publisher="Springer Science+Business Media, LLC",
address="New York",
year="2000"
}

@book{redner-2001,
author="S. Redner",
year="2001",
title="A Guide to First-Passage Processes",
address="Cambridge",
publisher="Cambridge University Press"
}

@article{dahlenburg_etal-jpa-2022,
title="Exact calculation of the mean first-passage time of 
continuous-time random walks by nonhomogeneous 
{Wiener--Hopf} integral equations",
author="M. Dahlenburg and G. Pagnini",
journal="J. Phys. A: Math. Theor.",
volume="55",
pages="505003",
year="2022",
note=""
}

@article{majumdar-pa-2010,
journal="Physica A",
volume="389",
year="2010",
pages="4299--4316",
title="Universal first-passage properties of discrete-time random walks 
and {L\'evy} flights on a line: {Statistics} of the global maximum and records",author="S. N. Majumdar"
}

@article{ivanov-aa-1994,
author="V. V. Ivanov", 
journal="Astron. Astrophys.",
volume="286",
year="1994",
pages="328--337",
title="Resolvent method: exact solutions of half-space transport problems
by elementary means"
}

@incollection{frisch_etal-1994,
author="U. Frisch and H. Frisch", 
title="Universality of escape from a half-space for symmetrical random walks",
booktitle="L\'evy Flights and Related Topics in Physics: 
Proceedings of the International Workshop Held at Nice, France, 
27--30 June 1994",
editor="M. F. Shlesinger and G. M. Zaslavsky and U. Frisch",
series="Lecture Notes in Physics", 
address="Berlin, Heidelberg",
publisher="Springer--Verlag", 
year="1994", 
pages="262--268"
}

@article{majumdar_etal-jsp-2006,
journal="J. Stat. Phys.",
volume="122", 
year="2006",
title="Unified solution of the expected maximum of a discrete time 
random walk and the discrete flux to a spherical trap",
author="S. N. Majumdar and A. Comtet and R. M. Ziff",
pages="833--856"
}

@article{zwanzig-jsp-1983,
title="From classical dynamics to continuous time random walks",
author="R. Zwanzig",
journal="J. Stat. Phys.",
volume="30",
year="1983",
pages="255--262"
}

@article{pollaczek-1952,
author="F. Pollaczek", 
title="Fonctions caract{\'e}ristiques de certaines r{\'e}partitions 
d{\'e}finies au moyen de la notion d'ordre. {Application} 
{\`a} la th{\'e}orie des attentes",
journal="C. R. Acad. Sci. Paris",
volume="234",
year="1952", 
pages="2334--2336",
}

@article{spitzer-1957,
title="The {Wiener--Hopf} equation whose kernel is a probability density",
author="F. Spitzer",
journal="Duke Math. J.",
volume="24",
pages="327--343",
year="1957"
}

@article{spitzer-1960,
title="The {Wiener--Hopf} equation whose kernel is a probability density. {II}",
author="F. Spitzer",
journal="Duke Math. J.",
volume="27",
pages="363--372",
year="1960"
}

@article{spitzer-1956,
journal="Trans. Am. Math. Soc.",
volume="82", 
pages="323--339",
year="1956",
title="A combinatorial lemma and its application to probability theory",
author="F. Spitzer"
}

@article{koren_etal-pa-2007,
journal="Physica A",
volume="379",
year="2007",
pages="10--22",
title="On the first passage time and leapover properties of {L\'evy} motions",
author="T. Koren and A. V. Chechkin and J. Klafter"
}

@article{koren_etal-prl-2007,
journal=prl,
volume="99", 
pages="160602",
year="2007",
title="Leapover lengths and first passage time statistics for 
{L\'evy} flights",
author="T. Koren and M. A. Lomholt and A. V. Chechkin and J. Klafter
and R. Metzler"
}

@article{sparre-andersen-1949,
title="On the number of positive sums of random variables",
author="E. {Sparre Andersen}",
journal="Scand. Actuar. J.",
volume="1",
year="1949",
pages="27--36"
}

@article{sparre-andersen-1953,
title="On the fluctuations af sums of random variables",
author="E. {Sparre Andersen}",
journal="Math. Scand.",
volume="1",
year="1953",
pages="263--285"
}

@article{sparre-andersen-1954a,
title="Remarks to the paper: 
On the fluctuations of sums of random variables",
author="E. {Sparre Andersen}",
journal="Math. Scand.",
volume="2",
year="1954",
pages="193--194"
}

@article{sparre-andersen-1954b,
title="On the fluctuations af sums of random variables {II}",
author="E. {Sparre Andersen}",
journal="Math. Scand.",
volume="2",
year="1954",
pages="195--223"
}

@article{lindley-1952,
author="D. V. Lindley", 
title="The theory of queues with a single server", 
journal="Math. Proc. Cambridge Philos. Soc.",
volume="48",
year="1952", 
pages="277--289"
}

@Inbook{aurzada_etal-2015,
author="F. Aurzada and T. Simon",
title="Persistence Probabilities and Exponents",
booktitle="L{\'e}vy Matters V: Functionals of L{\'e}vy Processes",
year="2015",
publisher="Springer International Publishing",
address="Cham",
pages="183--224"
}

@article{frisch_etal-1975,
author="U. Frisch and H. Frisch", 
year="1975",
title="{Non-LTE} transfer. $\sqrt{\epsilon}$ revisited", 
journal="Mon. Not. R. Astr. Soc.",
volume="173", 
pages="167--182"
}

@article{dahlenburg_etal-prsa-2023,
journal="Proc. R. Soc. A",
volume="479", 
year="2023",
title="{Sturm--Liouville} systems for the survival probability in first-passage time problems",
author="M. Dahlenburg and G. Pagnini",
pages="20230485"
}

@article{artuso_etal-pre-2014,
journal="Phys. Rev. E",
volume="89", 
pages="052111",
year="2014",
title="{Sparre--Andersen} theorem with spatiotemporal correlations",
author="R. Artuso and G. Cristadoro and M. {Degli Esposti}
and G. Knight"
}

@article{lacroix_etal-jpa-2020,
title="Universal survival probability for a correlated random walk 
and applications to records",
author="B. {Lacroix-A-Chez-Toine} and F. Mori",
year="2020",
journal="J. Phys. A: Math. Theor.",
volume="53",
pages="495002"
}

@article{dybiec_etal-jpa-2016,
title="To hit or to pass it over--remarkable transient behavior 
of first arrivals and passages for {L{\'e}vy} flights in finite domains",
author="B. Dybiec and E. {Gudowska--Nowak} and A. Chechkin",
year="2016",
journal="J. Phys. A: Math. Theor.",
volume="49",
pages="504001"
}

@article{efros-1935,
author="A. M. Efros", 
title="Some applications of operational calculus in analysis",
journal="Mat. Sb.", 
volume="42",
year="1935",
pages="699--705"
}

@article{wlodarski-1952,
author="L. W{\l{}}odarski", 
title="Sur une formule de {Efros}",
journal="Stud. Math.",
year="1952", 
volume="13", 
pages="183--187"
}

@article{shkilev-pre-2024,
journal=pre,
volume="110", 
pages="024139",
year="2024",
title="First-passage and first-arrival problems in 
continuous-time random walks: {Beyond} the diffusion approximation",
author="V. P. Shkilev"
}

@article{ding_etal-pre-1995,
title="Distribution of the first return time in {fractional Brownian motion} 
and its application to the study of on-off intermittency",
author="M. Ding and W. Yang",
journal=pre,
year="1995",
volume="52",
pages="207--213"
}

@article{colaiori_etal-pre-2004,
journal=pre,
volume="69", 
pages="041105",
year="2004",
title="Average trajectory of returning walks",
author="F. Colaiori and A. Baldassarri and C. Castellano"
}

@incollection{jeon_etal-book-2014,
title="First passage behavior of multi-dimensional {fractional Brownian motion} 
and application to reaction phenomena",
author="J.-H. Jeon and A. V. Chechkin and R. Metzler",
booktitle="First-Passage Phenomena and Their Applications",
publisher="World Scientific", 
year="2014", 
pages="175--202"
}

@article{gruda_etal-jsm-2025,
title="The joint distribution of first return times and 
of the number of distinct sites visited by a {1D} random walk before 
returning to the origin",
author="M. Gruda and O. Biham and E. Katzav and R. K{\"u}hn",
journal="J. Stat. Mech.",
year="2025",
pages="013203"
}

@article{kostinski_etal-ajp-2016,
author="S. Kostinski and A. Amir",
year="2016",
title="An elementary derivation of first and last return times of 
{1D} random walks",
journal="Am. J. Phys.",
volume="84",
pages="57--60"
}

\end{document}